%% file: main.tex
\documentclass[conference]{IEEEtran}
\IEEEoverridecommandlockouts

\usepackage{cite}
\usepackage{amsmath,amssymb,amsfonts}
\usepackage{algorithmic}
\usepackage{graphicx}
\usepackage{textcomp}
\usepackage{xcolor}
\usepackage{amsthm}
\usepackage[utf8]{inputenc}
\usepackage[nolist]{acronym}
\usepackage{subcaption}

\usepackage{subcaption}
\usepackage[percent]{overpic}
\usepackage{xurl}
\usepackage[hidelinks]{hyperref}
\newtheorem{lemma}{Lemma}
\def\BibTeX{{\rm B\kern-.05em{\sc i\kern-.025em b}\kern-.08em
    T\kern-.1667em\lower.7ex\hbox{E}\kern-.125emX}}
    
\begin{document}
\bstctlcite{IEEEexample:BSTcontrol}
\input{AcronymList.tex}
\title{VR-VFL: Joint Rate and Client Selection for Vehicular Federated Learning Under Imperfect CSI  \vspace{-0.25cm}
}

\author{Metehan Karatas, Subhrakanti Dey, Christian Rohner, Jos\'{e} Mairton Barros da Silva Júnior\\
Uppsala University, Uppsala, Sweden \\ Email: \{metehan.karatas, christian.rohner, mairton.barros\}@it.uu.se, subhrakanti.dey@angstrom.uu.se 
\thanks{The computational and data processing tasks were enabled by resources provided by the National Academic Infrastructure for Supercomputing in Sweden (NAISS), partially funded by the Swedish Research Council through grant agreement no. 2022-06725. The work of C. Rohner is partly supported by the Swedish Research Council through grant agreement no. 2024-05758.} \vspace{-0.55cm}
}

\maketitle

\begin{abstract}
    Federated learning in vehicular edge networks faces major challenges in efficient resource allocation, largely due to high vehicle mobility and the presence of imperfect channel state information. Many existing methods oversimplify these realities, often assuming fixed communication rounds or ideal channel conditions, which limits their effectiveness in real-world scenarios. To address this, we propose variable rate vehicular federated learning (VR-VFL), a novel federated learning method designed specifically for vehicular networks under imperfect channel state information. VR-VFL combines dynamic client selection with adaptive transmission rate selection, while also allowing round times to flex in response to changing wireless conditions. At its core, VR-VFL is built on a bi-objective optimization framework that strikes a balance between improving learning convergence and minimizing the time required to complete each round. By accounting for both the challenges of mobility and realistic wireless constraints, VR-VFL offers a more practical and efficient approach to federated learning in vehicular edge networks. Simulation results show that the proposed VR-VFL scheme achieves convergence approximately 40\% faster than other methods in the literature.
\end{abstract}

\vspace{-0.30cm}
\section{Introduction} \vspace{-0.1cm}

\Ac{FL} enables collaborative machine learning model training across multiple devices without the need to communicate raw data \cite{mcmahanCommunicationEfficientLearningDeep2017}. Each \ac{FL} training iteration, termed a round, consists of local model updates at devices, local model updates are communicated to a server that aggregates a global model, and sends it back to the local nodes with the round time dictating how long this process takes.

Owing to its inherently distributed architecture, federated learning is well-suited for deployment in edge networks; however, FL still faces inherent challenges, especially in wireless environments~\cite{hellstromWirelessMachineLearning2022}. The communication overhead associated with frequent model exchanges between devices and a central server remains a significant constraint \cite{mcmahanCommunicationEfficientLearningDeep2017}. Moreover, the heterogeneity of edge devices in terms of data distributions, often \ac{non-IID} \cite{zhaoFederatedLearningNonIID2018}, and their computational and communication capabilities complicate the \ac{FL} training and impact model performance. 

A particularly challenging and relevant application domain is \ac{FLVEN} \cite{posnerFederatedLearningVehicular2021}, leveraging vehicles as mobile learning clients. The inherent high mobility in vehicular networks exacerbates challenges related to client availability and communication reliability due to high mobility induced rapid channel variations via Doppler effects, severely limiting the accuracy and timeliness of \ac{CSI} acquisition \cite{liResourceAllocationD2DBased2020a}. Due to these reasons, operating under \ac{I-CSI} is the norm. Relying on \ac{I-CSI} complicates resource management decisions in \ac{FLVEN}, such as client scheduling and rate allocation, undermining the potential for efficient FL deployment.

Although some studies incorporated channel uncertainty~\cite{salehiFederatedLearningUnreliable2021} or analyzed the effects of \ac{I-CSI}~\cite{paseConvergenceTimeFederated2021}, significant gaps persist, particularly in the context of \ac{FLVEN}. While \cite{salehiFederatedLearningUnreliable2021} presented an effective framework for FL under channel uncertainty, it did not account for \ac{I-CSI} or client mobility. Additionally, its reliance on multiple retransmissions neglects the resulting delay in global convergence. The work in~\cite{paseConvergenceTimeFederated2021}, while insightful for general \ac{FL}, relied on simplifying assumptions such as identical \ac{SINR} across devices, achieved through compensating transmit powers, and optimistic capacity achieving transmission schemes even with imperfect CSI, potentially overestimating participation from users with poor average channel conditions. It also assumes a statistical CSI model rather than true instantaneous \ac{I-CSI}, using the channel's distributional properties without incorporating estimation mechanisms or errors.

Similarly, authors in \cite{waduFederatedLearningChannel2020} employed a Lyapunov framework for scheduling under \ac{I-CSI}, but lack key vehicular aspects such as mobility and large-scale fading variations. In \cite{zhouDynamicResourceManagement2024}, despite using an accurate \ac{I-CSI} model, its primary target is maximum client participation, without any discussion on data heterogeneity or its impact on convergence quality. Moreover, a critical modeling detail often overlooked in \ac{FL} works under \ac{I-CSI}~\cite{paseConvergenceTimeFederated2021, zhangVehicleSelectionResource2024} is that the channel estimation error acts as interference in the \ac{SINR} expression rather than contributing to the useful signal power, an effect only explicitly considered in a few studies such as~\cite{zhouDynamicResourceManagement2024}.

\Ac{FLVEN} has received significant research attention in areas such as resource allocation and scheduling. However, vehicular communication challenges under realistic channel conditions remain largely unexplored. The work in \cite{zhangVehicleSelectionResource2024} has neglected \ac{I-CSI} altogether while employing fixed round times. The authors in \cite{yanDynamicSchedulingVehicletoVehicle2024} have not considered \ac{I-CSI} effects and data heterogeneity in their problem formulation. More broadly, the related works~\cite{salehiFederatedLearningUnreliable2021}, \cite{paseConvergenceTimeFederated2021}, \cite{waduFederatedLearningChannel2020}, \cite{zhouDynamicResourceManagement2024}, \cite{zhangVehicleSelectionResource2024} and \cite{yanDynamicSchedulingVehicletoVehicle2024} continue to assume fixed round times per FL iteration. Such rigidity fails to adapt to the dynamic channel conditions and computational heterogeneity inherent in edge networks, where flexible time and resource allocation per round is crucial for efficient and accurate learning.  
\begin{figure}[t!]
    \centerline{\includegraphics[width=0.39 \textwidth]{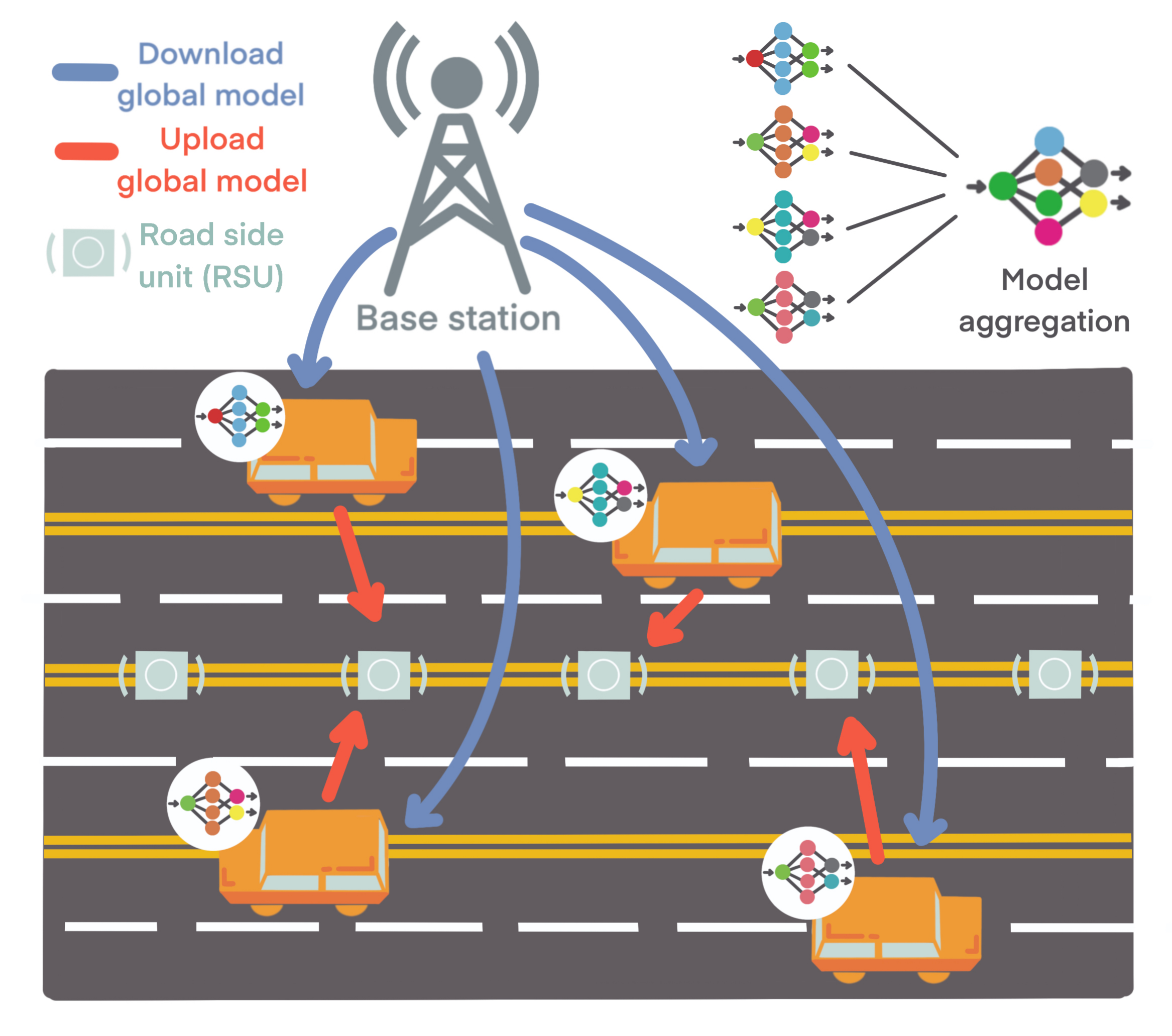} \vspace{-0.2cm}}
    \caption{The \ac{FL} operation, where four vehicles update their locally trained models to the RSUs, then the base station aggregates the models and broadcasts the new global model.} \vspace{-0.45cm}
    \label{fig}
\end{figure} 
To bridge these gaps, we propose a novel joint rate selection and client scheduling algorithm, termed \ac{VR-VFL}, designed for \ac{FLVEN} systems operating under \ac{I-CSI} due to mobility. Our main contributions in this work are: \vspace{-0.05cm}
\begin{itemize}
\item In Section~\ref{sec:systemModel}, we model time correlated stochastic \ac{I-CSI} channel using a Gauss-Markov process and integrate it into the \ac{SINR} expression for \ac{FLVEN}; a previously unexplored integration that enables accurate modeling of mobility induced interference effects in such scenarios.
\item In Section~\ref{clientParticipationRateSelection}, we propose a novel rate selection and client scheduling problem for \ac{FLVEN} under realistic \ac{I-CSI}. We adaptively select individual transmission rates for chosen vehicles based on their wireless and learning conditions, enabling control over the round duration in each communication round.
\item To this end, we introduce a novel bi-objective optimization problem that balances learning convergence of \ac{FL} with round time. We prove that the optimization problem is bi-convex, and propose a \ac{BCD} algorithm to solve the problem efficiently while proving that the solution converges to a stationary point.
\item In Section~\ref{sec:numericalResults}, our results show that \ac{VR-VFL} enables significant time savings, completing 1000 rounds in approximately 64k seconds for both IID and non-IID settings, achieving $73\%$ and $60\%$ test accuracy respectively, with the CIFAR-10 dataset, while the schemes from \cite{salehiFederatedLearningUnreliable2021} take around 110k seconds to reach the same accuracy. Notably, \ac{VR-VFL} achieves comparable performance to these schemes with $42\%$ reduced training time. 
\end{itemize} \vspace{-0.15cm}
To the best of our knowledge, this is the first work in \ac{FLVEN} to jointly optimize client selection, per client rate allocation, and flexible round timing under a realistic \ac{I-CSI} model, thus moving beyond conventional resource allocation by directly controlling round duration through rate selection. The source code for \ac{VR-VFL} and our simulations are publicly available.\footnote{\url{https://github.com/Eshinos/ICC2026_VRVFL_Simulation}} \vspace{-0.2cm}

\section{System Model} \label{sec:systemModel}
\vspace{-0.0cm}
We consider a vehicular network located on a highway as shown in Fig.~\ref{fig}, consisting of $V$ vehicles and evenly spaced \acp{RSU}, to which the vehicles transmit their trained models at every global \ac{FL} round. Vehicles arrive according to a Poisson process with rate $\lambda$, enter network coverage area, and leave depending on their speed, which is constant and randomly selected between a predefined interval. \vspace{-0.5cm}
\subsection{Channel Model}
The channel between the vehicles and \acp{RSU} are affected by the Doppler effect due to high velocity of the vehicles. We denote the channel gain between vehicle~$v$ and its nearest \ac{RSU} as $g_{v}= |h_{v}|^2 L_{v}$ where $h_{v}$ is the fast fading channel component, and $L_{v}$ is the large scale channel gain incorporating path-loss and shadow fading. Due to high mobility of the vehicles, the channel gain is time varying and we model it as a first order Gauss-Markov process \cite{liResourceAllocationD2DBased2020a}. We consider the time-varying fast fading component as \vspace{-0.1cm}
\begin{equation}
    h_{v}=\epsilon_v \hat{h}_{v}+\sqrt{1-\epsilon_v^2} \tilde{h}_{v}, \vspace{-0.1cm}
\end{equation} 
where $\hat{h}_{v}$ is the minimum mean squared error estimate of the fast fading channel gain, $\epsilon_v$ is the temporal correlation coefficient for vehicle $v$, and $\tilde{h}_{v}$ is the channel estimation error term that is independent of $\hat{h}_{v}$. Note that we have $\hat{h}_{v}, \tilde{h}_{v} \sim \mathcal{CN}(0,1)$ for both. We define the temporal correlation coefficient as $\epsilon_v\triangleq J_0(2\pi f_{v}^s T)$, where $f_{v}^s$ is the Doppler frequency for vehicle $v$, defined as $f_{v}^s = \vartheta_v f_c/c$, with $\vartheta_v$ as the velocity of vehicle $v$, $f_c$ as the carrier frequency, $c$ as the speed of light, and $T$ as the channel feedback delay. Lastly, $J_0$ is the zeroth order Bessel function of the first kind.  \vspace{-0.2cm}

\subsection{Signal Model} \vspace{-0.05cm}

Vehicle $v$ transmits its local model through its wireless channel, $g_{v}$. We assume orthogonal transmissions across vehicles, so the received signal from vehicle $v$ is
\begin{equation}
    S_v=\underbrace{\sqrt{L_{v}P_v} \epsilon_v \hat{h}_v s_v}_{\text{Signal of Interest}} 
    + \underbrace{\sqrt{L_{v}P_v}\sqrt{1 - \epsilon_v^2} \tilde{h}_v s_v}_{\text{Estimation Error}} 
    + \underbrace{n_0}_{\text{Noise}},
\end{equation}
where $s_v$ is the unit-power transmitted symbol such that $\mathbb{E} \{ \lvert {s_v} \rvert ^2 \}=1$, $P_v$ is the transmission power of vehicle $v$, and noise term $n_0$ follows $\mathcal{CN}(0,N_0)$.

The transmission is negatively affected by the \ac{I-CSI}, i.e., the estimation error is not known and it cannot be exploited by the transmission. Then, the \ac{SINR} for vehicle $v$ is \vspace{-0.2cm}
\begin{equation} \label{eq:sinr}
    \gamma_v = \frac{P_vL_{v} \epsilon_v^2 |\hat{h}_{v}|^2}{W_v N_0 + P_vL_{v} (1-\epsilon_v^2)|\tilde{h}_{v}|^2},
\end{equation}
where $W_v N_0$ is the noise power with $W_v$ as the bandwidth allocated to vehicle $v$. We emphasize that the channel estimation error term appears in the denominator of the \ac{SINR} expression. This indicates that the channel estimation error cannot contribute as a useful signal.

Accordingly, the achievable capacity for vehicle $v$ is \vspace{-0.2cm}
\begin{equation} \label{eq:Cap} 
    C_v = W_v \log_2(1+\gamma_v). \vspace{-0.2cm}
\end{equation}
Since the receiver has \ac{I-CSI}, the achievable capacity is negatively affected by the channel estimation error. Note that this error scales with the transmission power $P_v$, implying that the capacity degradation cannot be solved by increasing $P_v$. \vspace{-0.2cm}

\subsection{Learning Model}
We denote the dataset of vehicle $v$ as $\mathcal{D}_v=\{ \mathcal{X}_v, \quad \mathcal{Y}_v \}$ and its cardinality as $D_v$. The training dataset for vehicle $v$ is $\mathcal{X}_v= \{ x_{v,1},x_{v,1},\dots x_{v,D_v}\}$ and the corresponding labels are $\mathcal{Y}_v= \{ y_{v,1},y_{v,1},\dots y_{v,D_v}\}$. The model parameters are $\boldsymbol{w}$ and the loss function for data sample $j$ of vehicle $v$ is $f(\boldsymbol{w}; x_{v,j}, y_{v,j})$. For each vehicle $v$, the local loss function is \vspace{-0.1cm}
\begin{equation}
    F_v(\boldsymbol{w})=\frac{1}{D_v}\sum\nolimits_{j=1}^{D_v}f(\boldsymbol{w}; x_{v,j}, y_{v,j}), \vspace{-0.1cm}
\end{equation}
and the global loss function is $F(\boldsymbol{w})=\frac{1}{D}\sum\nolimits_{v\in \mathcal{V}}D_v F_v(\boldsymbol{w})$ where $D=\sum_{v\in \mathcal{V}}D_v$ is the total number of data samples. Then, we formulate the \ac{FL} problem as finding the optimal model parameters $\boldsymbol{w}^*$ such that \vspace{-0.1cm}
\begin{equation} \label{eq:FLproblem}
    \boldsymbol{w}^* = \arg\min_{\boldsymbol{w}} F(\boldsymbol{w}). \vspace{-0.1cm}
\end{equation}

We assume that \ac{FL} training follows a round-based approach. In each round, the current global model is broadcasted to the vehicles. Vehicles participating in the round train this model with their local datasets, and transmit their trained models to the nearest \ac{RSU}. The received models are aggregated and the global model is updated as \vspace{-0.1cm}
\begin{equation} \label{eq:standardAggregation}
    \boldsymbol{w}_{t+1} = {\sum_{v \in \mathcal{V}_t}D_v \boldsymbol{w}_{t,v}} \Big/ {\sum_{v \in \mathcal{V}_t}D_v},\vspace{-0.1cm}
\end{equation} 
where $\boldsymbol{w}_{t,v}$ is the model of vehicle $v$ trained at round $t$, $\mathcal{V}_t$ is the set of vehicles that successfully participate in round $t$ and $\boldsymbol{w}_{t+1}$ is the updated global model. The training process continues until the global model converges. Note that, although some parameters may change between different rounds, we omit the round index $t$ for simplicity. \vspace{-0.15cm}
\subsection{Successful Transmission Condition} \label{subsection:successfulTransmissionCondition}
We denote $T_t$ as the allocated duration for global round $t$ in \ac{FL} training, during which the trained model, consisting of $Z$ bits, must be transmitted by participating vehicles. The sojourn time of a vehicle refers to the duration it remains within the coverage area. Then, any vehicle $v$ contributing in round $t$ must satisfy the inequality  \vspace{-0.10cm}
\begin{equation} \label{eq:requirement}
    C_v \geq R_{v,t}= \max \{ Z/T_t, Z/T_v^t \}, \vspace{-0.10cm}
\end{equation} 
where $R_{v,t}$ is the transmission rate of vehicle $v$ at round $t$ and $T_v^t$ is the remaining sojourn time, i.e., remaining distance in coverage divided by velocity. The minimum rate required for a vehicle to participate is given by $R_{v,t}^{\textrm{min}}=Z/T_t$, which ensures that the full model of size $Z$ can be transmitted within the round time $T_t$. On the other hand, the maximum transmission rate is defined as $R_{v,t}^{\textrm{max}}=C_v$ when $|\tilde{h}_{v}|^2=0$ according to Eqs. \eqref{eq:sinr} and \eqref{eq:Cap}, corresponding to the upper bound on the channel capacity that captures the best-case scenario of perfect \ac{CSI}. The condition in Eq. \eqref{eq:requirement} guarantees that a vehicle can transmit its local model within the allocated round time, taking into account both its transmission rate and remaining sojourn time. The round time $T_t$ itself is determined by the vehicle with the lowest transmission rate among all clients in the round $t$, and is given by $T_t = \min_{ v \in \mathcal{V}_t} \{ Z/R_{v,t} \} $. However, vehicles with remaining sojourn time, $T_v^t$, less than $T_t$ can still participate by transmitting at a higher rate to offset their reduced availability window. From Eqs. \eqref{eq:sinr} and \eqref{eq:Cap}, the condition in Eq. \eqref{eq:requirement} is equivalent to
\begin{equation} \label{eq:condition}
    |\hat{h}_{v}|^2 \geq \underbrace{\frac{\left( 2^\frac{R_{v,t}}{W_v} - 1 \right) \left(1-\epsilon_v^2\right)}{\epsilon_v^2}}_{a_v} |\tilde{h}_{v}|^2+\underbrace{\frac{\left( 2^\frac{R_{v,t}}{W_v} - 1 \right) W_v N_0}{P_vL_{v} \epsilon_v^2}}_{b_v}.
 \vspace{-0.15cm}\end{equation}

Note that for a given vehicle $v$ with given transmission rate $R_{v,t}$, $a_v$ and $b_v$ are constants related to the relative power of the estimation error and noise, respectively, with respect to the power of the signal of interest. The system has access to the estimate $\hat{h}_{v}$, and the randomness in Eq. \eqref{eq:condition} arises solely from the estimation error term $\tilde{h}_{v}$. We denote the condition in Eq. \eqref{eq:condition} as event $A_v$ and rewrite it as \vspace{-0.1cm}
\begin{equation} \label{eq:Av}
    A_v:|\tilde{h}_{v}|^2 \leq ({|\hat{h}_{v}|^2-b_v})/{a_v}. \vspace{-0.00cm}
\end{equation}

 We can evaluate the probability of successful transmission, event $A_v$, over the randomness of $|\tilde{h}_{v}|^2$. Then, the probability of successful transmission for vehicle $v$ is \vspace{-0.1cm}
\begin{equation} \label{eq:outageProbability}
    \mathbb{P}(A_v \mid \hat{h}_{v}) =
    \begin{cases}
        1 - \exp \left(-({|\hat{h}_{v}|^2 - b_v})/{a_v} \right), & \text{if } |\hat{h}_{v}|^2 > b_v, \\
        0, & \text{otherwise},
    \end{cases}
\end{equation}
from \ac{CDF} of the exponential random variable that is $|\tilde{h}_{v}|^2$. \vspace{-0.00cm}

\section{Client Participation and Rate Selection} \label{clientParticipationRateSelection} \vspace{-0.0cm}
\subsection{Convergence Analysis} \vspace{-0.0cm}
The convergence analysis of FL over wireless networks has been extensively studied in prior works, such as~\cite{chenJointLearningCommunications2021}. However, the convergence analysis of \ac{FLVEN} differs from this standard, especially when I-CSI is present. This distinction arises due to several factors: limited participation from vehicles due to different sojourn times, dynamic number of clients in the system, non i.i.d. data distribution among vehicles, different round times, and severity of \ac{I-CSI}. 

These challenges directly affect the  convergence of the \ac{FL} problem in Eq. \eqref{eq:FLproblem}. Specifically in our \ac{FLVEN} setting, the expected decrease in the global objective at round $t$, with respect to the global minimum, $F^*$, i.e., $\mathbb{E}[F(\boldsymbol{w_t})-F^*]$, is influenced by the unreliable participation and transmission conditions of vehicles. In~\cite{salehiFederatedLearningUnreliable2021}, it is shown that the upper bound on $\mathbb{E}[F(\boldsymbol{w_t})-F^*]$ depends on key parameters formulated as
\begin{equation} \label{eq:objectiveParameter}
    \mathbb{E}[F(\boldsymbol{w_t})-F^*] \propto \sum\nolimits_{v=1}^{V} \frac{D_v}{D} \left(\frac{1}{u_{v,t}\mathbb{P}(A_v|\hat{h}_{v})} -1\right).
\end{equation}
where $u_{v,t}$ is the probability of including vehicle $v$ in round $t$, and we define it as $u_{v,t} =\mathbb{E} ( \mathbf{1}(v \in \mathcal{V}_t))$, where $\mathbf{1}(.)$ is the indicator function for the event in its argument.

To address these challenges, we construct a similar argument to \cite{zhangVehicleSelectionResource2024}, which relies on the convergence framework of \cite{salehiFederatedLearningUnreliable2021},
 and adopt Scheme~2 from \cite{salehiFederatedLearningUnreliable2021}, which allocates $u_{v,t}$ resource blocks to vehicle $v$ on average. Then, the aggregation in Eq. \eqref{eq:standardAggregation} needs to be modified to reflect the inclusion and successful transmission probabilities of the vehicles, i.e.,
\begin{equation}
    \boldsymbol{w}_{t+1} = \sum\nolimits_{v \in \mathcal{V}_t}\frac{D_v }{D} \frac{\boldsymbol{w}_{t,v} \mathbf{1}(v \in \mathcal{V}_t, A_v) }{u_{v,t}\mathbb{P}(A_v|\hat{h}_{v})}.
\end{equation} 

These modifications aim to amplify the contributions of vehicles that are less frequently selected or experience lower transmission success rates. This ensures that vehicles with unreliable connectivity or limited sojourn times still exert meaningful influence on model training, preventing bias toward frequently available clients. \vspace{-0.15cm}

\subsection{Problem Formulation}\vspace{-0.10cm}
In each round, our goal is to minimize the upper bound on the expected reduction of the loss function while keeping the round time as short as possible. We propose a novel bi-objective optimization problem with variable round time as
\begin{subequations} \label{eq:optimizationProblem2} \vspace{-0.3cm}
    \begin{align}
        \mathcal{P}: \min_{\boldsymbol{u}_{t}, \mathbf{R}_t} \quad &  \sum\nolimits_{v \in \mathcal{V}^{\textrm{feas}}_t} \frac{\alpha D_v}{ D u_{v,t}\mathbb{P}(A_v|\hat{h}_{v})} \nonumber \\&+ (1-\alpha) \max_v \left(u_{v,t} \exp (-( 2^\frac{R_{v,t}}{W_v} - 1 ))\right)  \label{eq:objective2}\\
        & \sum\nolimits_{v \in \mathcal{V}^{\textrm{feas}}_t} u_{v,t} \le N, \label{eq:irbConstraint2} \\
        & {u_{\textrm{min}} \le u_{v,t} \le 1}, \quad \textrm{for} \quad v \in \mathcal{V}^{\textrm{feas}}_t, \label{eq:inclusionConstraint2} \\
        & R_{v,t}^{\textrm{min}} \le R_{v,t} \le R_{v,t}^{\textrm{max}}, \quad \textrm{for} \quad v \in \mathcal{V}^{\textrm{feas}}_t, \label{eq:rateConstraint2} 
     \end{align}
\end{subequations}
where $R_{v,t}^{\textrm{min}}$ and $R_{v,t}^{\textrm{max}}$ are defined in Section~\ref{subsection:successfulTransmissionCondition}, the dependence on $R_{v,t}$ of $\mathbb{P}(A_v|\hat{h}_{v})$, $a_v$ and $b_v$ is shown in Eqs. \eqref{eq:condition} and \eqref{eq:outageProbability}, $\mathcal{V}^{\textrm{feas}}_t$ is the set of feasible vehicles for round $t$ such that $\mathcal{V}^{\textrm{feas}}_t = \{ v | R_{v,t}^{\textrm{min}} < R_{v,t}^{\textrm{max}} \}$, $\mathbf{u}_t=[u_{1,t}, u_{2,t},\dots , u_{V,t}]$, $\mathbf{R}_t=[R_{1,t}, R_{2,t},\dots , R_{V,t}]$ and $\alpha \in [0,1] $ controls the emphasis between model convergence (first term) and round time (second term). Constraints~\eqref{eq:irbConstraint2} and \eqref{eq:inclusionConstraint2} make sure that we allocate at most $N$ resource blocks and every vehicle gets at most one resource block, respectively. Constraint \eqref{eq:rateConstraint2} ensures that vehicles transmit below their upper capacity limit and above a minimum rate that guarantees completion before their sojourn time expires.

The first objective in Eq. \eqref{eq:objective2} aims at including as many vehicles as possible in round $t$, following an approach similar to those in~\cite{zhangVehicleSelectionResource2024} and~\cite{salehiFederatedLearningUnreliable2021}. This is equivalent to minimizing the objective defined in~\eqref{eq:objectiveParameter}. Conversely, the second objective in Eq.~\eqref{eq:objective2} aims at making the round as fast as possible by penalizing the lowest rate among clients which sets the round time. Note that, the second term shrinks with higher rates, encouraging faster rounds. Hence, the bi-objective optimization problem~\eqref{eq:optimizationProblem2} balances between taking shorter rounds that yield smaller improvements and longer rounds that result in larger loss function reductions. \vspace{-0.15cm}

\subsection{Solution to Problem $\mathcal{P}$} 

Directly solving the problem $\mathcal{P}$ is intractable due to the intricate coupling between the variables $\mathbf{R}_t$ and $\mathbf{u}_t$. While terms for different vehicles are separated in the objective function, at each term, $u_{v,t}$ is multiplied with a function of $R_{v,t}$ in the denominator of the first term in Eq.~\eqref{eq:objective2}, making the objective function non-convex jointly with respect to both variables. To address this challenge, we adopt a \ac{BCD} approach \cite{wrightCoordinateDescentAlgorithms2015}. \ac{BCD} is an iterative optimization algorithm used to solve problems involving multiple variables. Instead of optimizing all variables at once, which can be hard or computationally expensive, BCD optimizes one block (or group) of variables at a time while keeping the others fixed. It repeats this iterative process until convergence. As shown in~\cite{gorskiBiconvexSetsOptimization2007}, the \ac{BCD} algorithm is guaranteed to converge to a stationary point that is partially optimal, provided the objective function is bi-convex. Partial optimality means that, for a fixed vector $u^{v,t}$, no choice of $R^{t}$ can yield a better objective value, and vice versa. This is the strongest optimality guarantee that can generally be expected for non-convex, and in particular biconvex, optimization problems~\cite{gorskiBiconvexSetsOptimization2007}.

In our solution, we alternate between two steps in each iteration: first, fixing $\mathbf{u}_t$ and optimizing over $\mathbf{R}_t$; then, fixing $\mathbf{R}_t$ and optimizing over $\mathbf{u}_t$. This alternating procedure is repeated until a stationary point is reached. It is straightforward to verify that the objective function in problem $\mathcal{P}$ is convex with respect to $\mathbf{u}_t$ when $\mathbf{R}_t$ is fixed. The first term is a sum of functions of the form $1/u_{v,t}$, which are convex for positive $u_{v,t}$. The second term is a pointwise maximum over convex functions, the function being $u_{v,t}$ itself, and is therefore also convex. Since the sum of convex functions remains convex, the overall objective is convex in $\mathbf{u}_t$.

The convexity of the objective function with respect to $\mathbf{R}_t$ when $\mathbf{u}_t$ is fixed is established in Lemma~\ref{lemma:objectiveFunctionConvexity}. \vspace{-0.1cm}

\begin{lemma} \label{lemma:objectiveFunctionConvexity}
    Consider the optimization problem $\mathcal{P}$ given in \eqref{eq:optimizationProblem2}. The objective function in Eq. \eqref{eq:objective2} is convex with respect to the variables $\mathbf{R}_t$. \vspace{-0.1cm}
\end{lemma}
\begin{proof}
    Proof in Appendix~\ref{appendix:objectiveFunctionConvexity}. \vspace{-0.2cm}
\end{proof}

Since our objective is bi-convex in $\mathbf{R}_t$ and $\mathbf{u}_t$, the proposed \ac{BCD} approach is guaranteed to converge to a partially optimal solution for problem~$\mathcal{P}$. The proposed \ac{BCD} method is formulated and solved using the CVXPY framework~\cite{diamond2016cvxpy}. 

We denote our solution to problem~$\mathcal{P}$ as \ac{VR-VFL}, which selects clients and their transmission rates in each round. \vspace{-0.20cm}

\section{Simulation Results} \label{sec:numericalResults}\vspace{-0.10cm}
\subsection{Simulation parameters} \vspace{-0.10cm}
The simulation environment considers a highway scenario \cite{3GPP_TR_36885} with a simulated road segment of 2 km in length and 6 traffic lanes, each with a width of 4m. \acp{RSU} are deployed evenly, 100m apart, along the middle of the road segment.

The wireless channel considers LOS pathloss and shadowing \cite{kyosti2007winner}, with a carrier frequency of 5.9 GHz, a bandwidth of 10 MHz utilizing 20 resource blocks, a noise density of -174 dBm/Hz, and a channel feedback delay of $T=0.5$ ms. Vehicle traffic in each lane is generated as a Poisson process with an arrival rate parameter of 0.2. Each vehicle is assigned a speed drawn uniformly at random between 60-100 km/h. The transmit power for each vehicle is 23 dBm, and vehicles maintain their selected rates for the duration of the round.

\begin{figure*}[!t] 
    \centering
    \begin{minipage}[t]{0.45\textwidth}
        \centering
        \includegraphics[width=\linewidth]{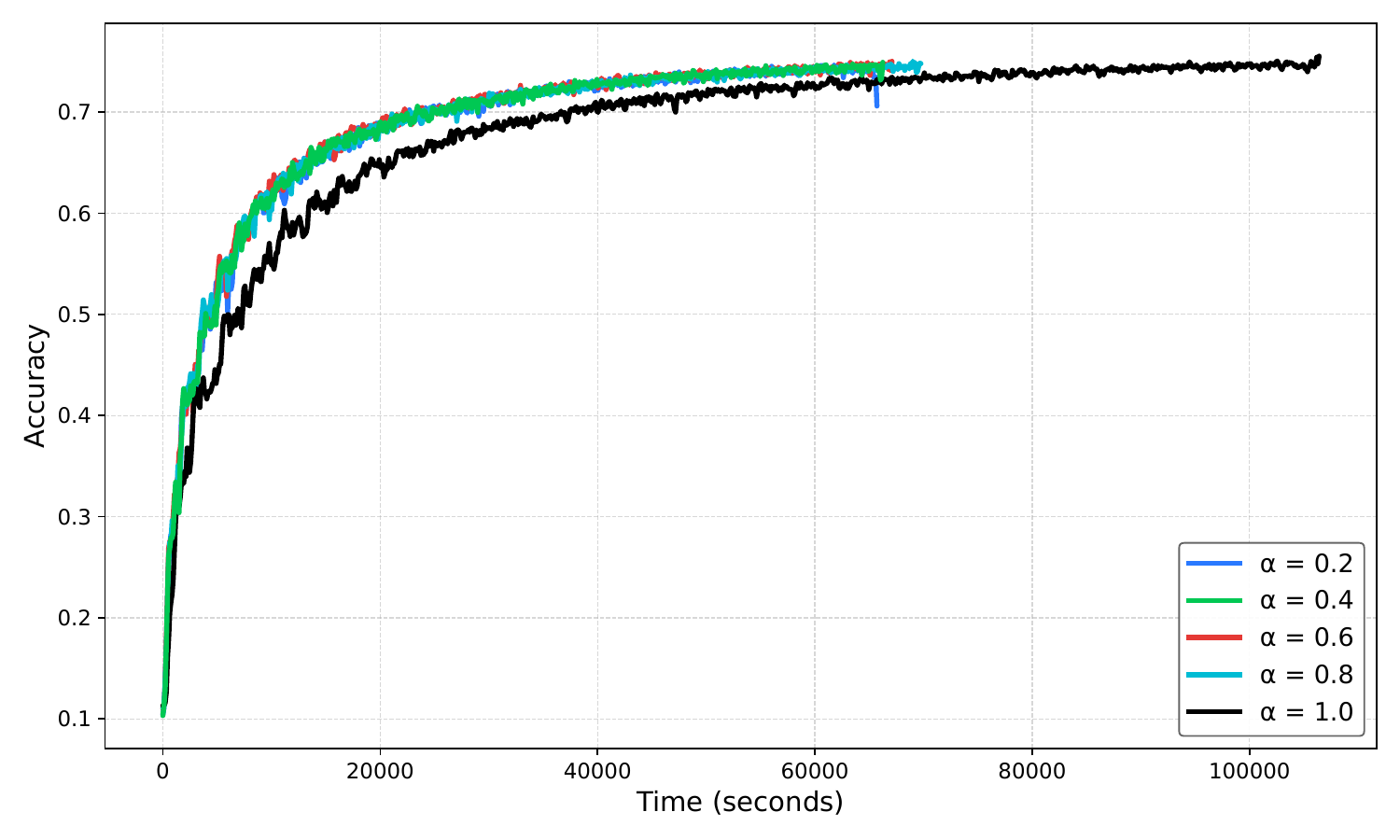}
        \subcaption{Test accuracy for \ac{VR-VFL} for different $\alpha$  (IID).}
        \label{fig:vrvfl_iid}
    \end{minipage}
    \hfill 
    \begin{minipage}[t]{0.45\textwidth}
        \centering
        \includegraphics[width=\linewidth]{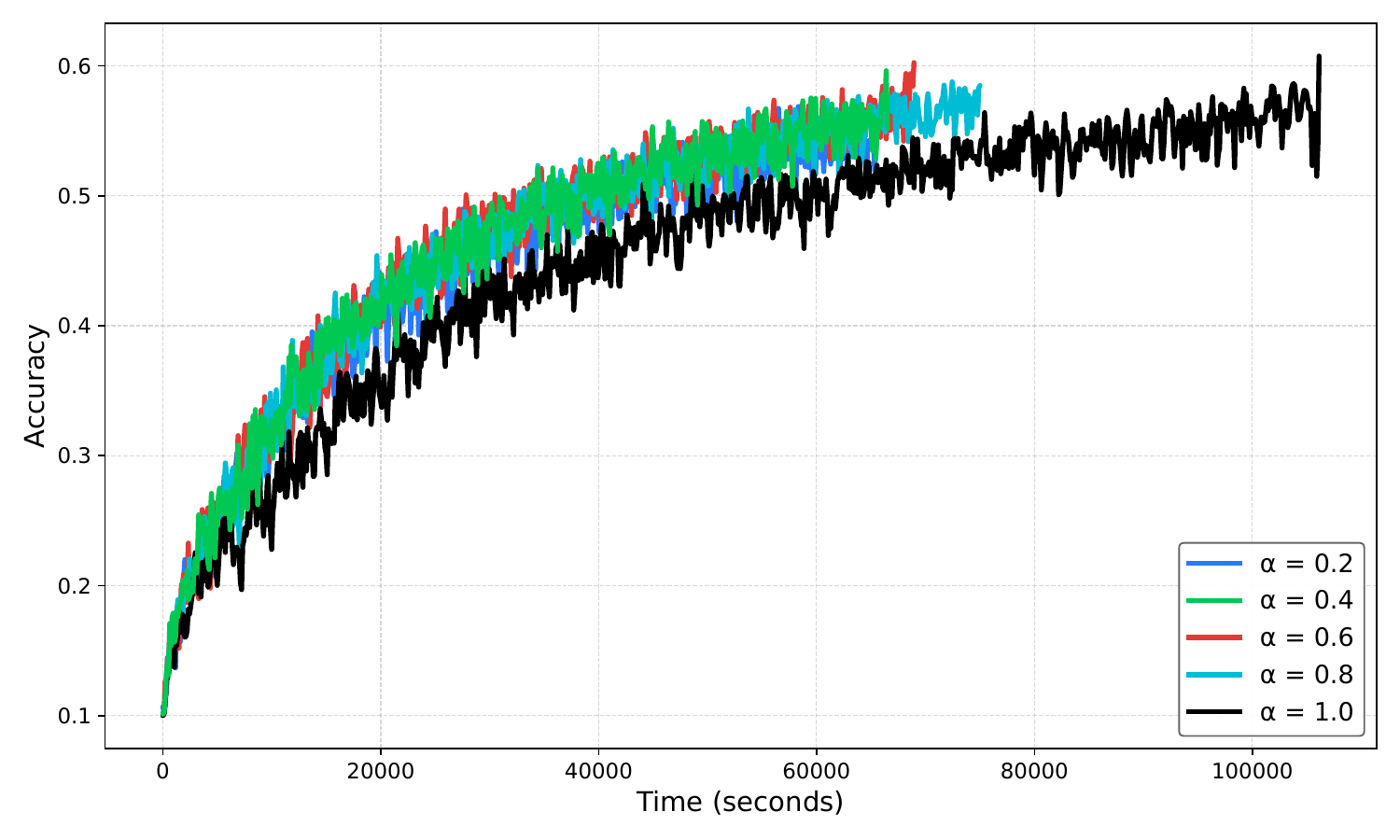}
        \subcaption{Test accuracy for \ac{VR-VFL} for different $\alpha$  (non-IID).}
        \label{fig:vrvfl_noniid}
    \end{minipage}
    \vspace{0.0em} 
    \begin{minipage}[t]{0.45\textwidth}
        \centering
        \includegraphics[width=\linewidth]{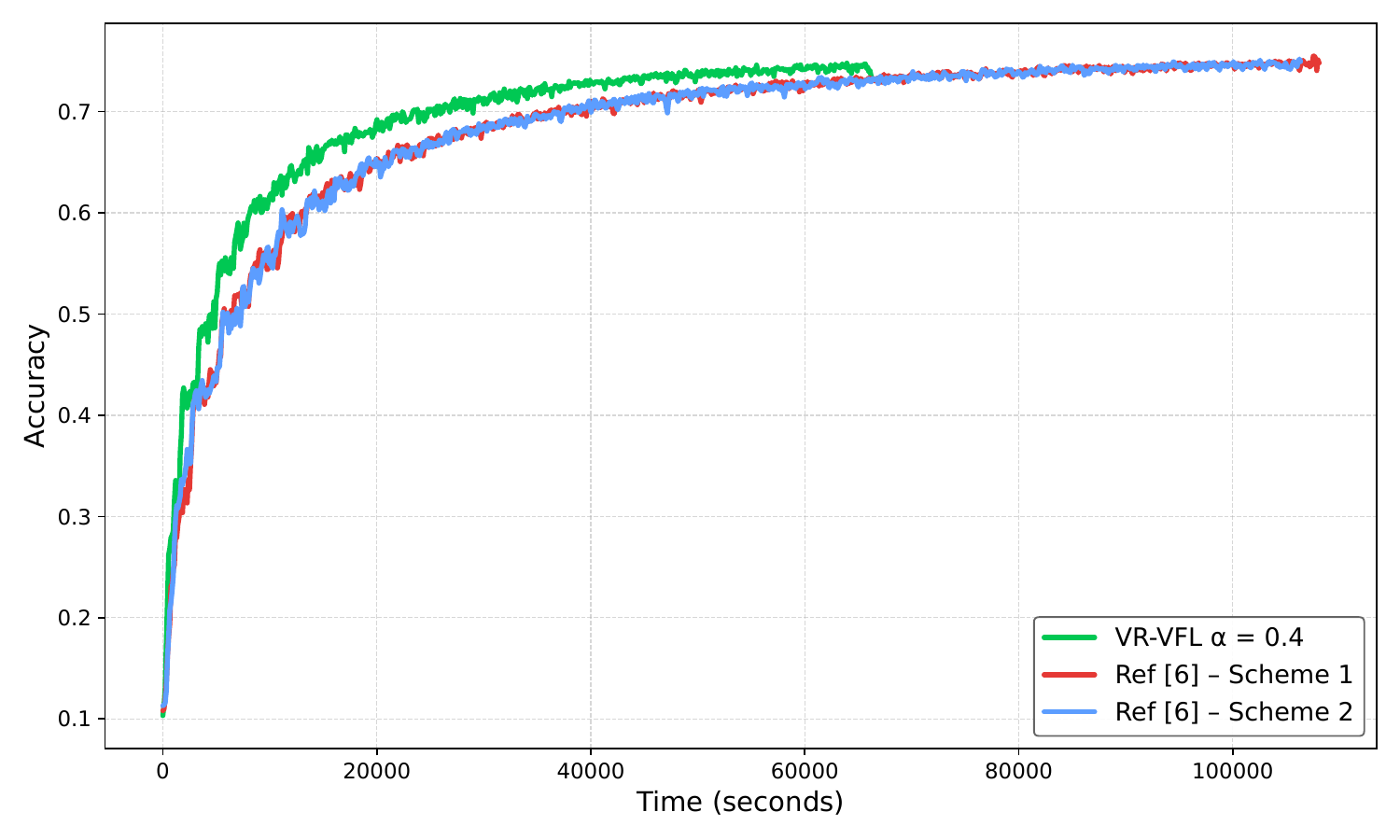}
        \subcaption{Test accuracy for \ac{VR-VFL}, Scheme 1 [6] and Scheme 2 [6].}
        \label{fig:comp_iid}
    \end{minipage}
    \hfill 
    \begin{minipage}[t]{0.45\textwidth}
        \centering
        \includegraphics[width=\linewidth]{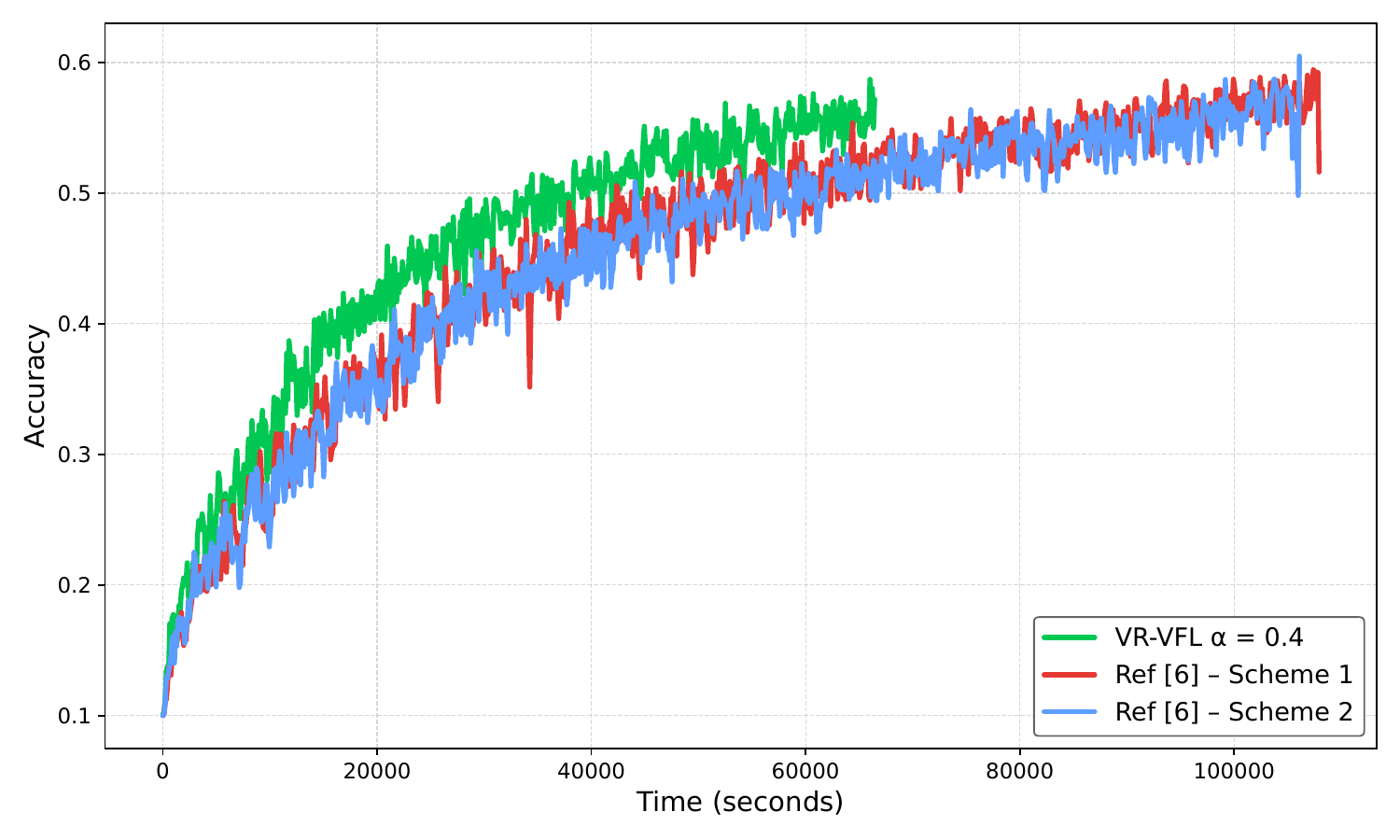}
        \subcaption{Test accuracy for \ac{VR-VFL}, Scheme 1 [6] and Scheme 2 [6].}
        \label{fig:comp_noniid}
    \end{minipage}

    \caption{ Test accuracy vs training time (s) for \ac{VR-VFL} under different $\alpha$: (a) IID, (b) non-IID. Test accuracy comparison for VR-VFL and Schemes 1-2 [6]: (c) IID, (d) non-IID. Results are averaged over 10 independent wireless simulations.}\vspace{-0.3cm}
    \label{fig:2x2_grid} 
\end{figure*} \vspace{-0.0cm}

For the learning model, we consider a scenario where vehicles dynamically enter and exit the region, while the model at the base station remains active and is continuously updated. Image classification is used as a benchmark, based on the CIFAR-10 dataset \cite{cifar10}, leaving more complex vehicular tasks, such as motion prediction, for future work. We use a decaying learning rate as $0.1/(1+\lfloor t/25 \rfloor)$, a momentum of 0.9, a FedProx parameter of 0.0025, a batch size of 32, and 5 local epochs. In the \ac{IID} setting, vehicles have 500 samples from each class. In the \ac{non-IID} setting, vehicles have between 2500 and 7500 samples selected uniformly at random from 1 to 3 randomly selected classes. The model size is 4.38 Mbits. The model consists of an initial convolutional layer with 3 input and 32 output channels, followed by two residual blocks maintaining the number of channels (32 in the first, 64 in the second). Downsampling is performed after each residual block using convolutional layers with stride 2, increasing the channels from 32 to 64 and then from 64 to 128, with skip connections. An adaptive average pooling layer and a classifier with two linear layers follow: the first maps 128 to 128 features with ReLU and dropout (0.1), and the second maps 128 to the 10 output classes.
\vspace{-0.04cm} 

Two schemes from \cite{salehiFederatedLearningUnreliable2021} are used for comparison. Scheme 1 from \cite{salehiFederatedLearningUnreliable2021} uniformly randomly chooses clients each round. Scheme 2 from \cite{salehiFederatedLearningUnreliable2021} solves the problem in \eqref{eq:optimizationProblem2} without the second term, i.e., the special case of VR-VFL with $\alpha=1$. Both schemes select the transmission rates with VR-VFL rate selection algorithm. Note that the method in \cite{paseConvergenceTimeFederated2021} is not used for comparison, as it assumes equal SNR and statistical channel across all clients, which is not feasible in our scenario.
\vspace{-0.2cm}
\subsection{Numerical Results} \vspace{-0.10cm}
The performance is evaluated based on test accuracy versus training time (accumulating round times) and all methods are run for 1000 FL rounds. Figures \ref{fig:vrvfl_iid} and \ref{fig:vrvfl_noniid} show performance of \ac{VR-VFL} for IID and non-IID distributions, respectively, for different $\alpha$. We notice that, decreasing $\alpha$ significantly accelerates convergence. This improvement stems from \ac{VR-VFL} placing a greater emphasis on the second objective in Eq.~\eqref{eq:objective2}, leading to slower clients being selected less frequently. Concurrently, the rates of the slowest clients are adaptively increased to mitigate their impact on the training time. Due to its fast and accurate convergence, converging at 63k seconds compared to 70k and 75k seconds for higher alphas, we set $\alpha=0.4$ for the rest of the experiments. \vspace{-0.00cm}

Figures \ref{fig:comp_iid} and \ref{fig:comp_noniid} compare the performance of \ac{VR-VFL} against Scheme 1 and Scheme 2 from \cite{salehiFederatedLearningUnreliable2021} for both IID and non-IID data. \Ac{VR-VFL} finishes its rounds in $43\%$ less time, achieving convergence in approximately 63k seconds for non-IID settings, while the other two take around 110k seconds to reach the same test accuracy. Note that, Scheme 1 and Scheme 2 show comparable performance, with Scheme 2 being only $1-2\%$ faster. This shows the importance of including the second term in \eqref{eq:optimizationProblem2}, which pushes for higher rates in slower clients, thus reducing round time.
\vspace{-0.2cm}
\section{Conclusion}  \vspace{-0.15cm}
In this paper, we investigated the challenges of \ac{FLVEN} by first employing a practical Gauss-Markov channel model tailored for this dynamic environment. Building upon this, we formulated an optimization problem aimed at minimizing the FL loss while accounting for the unique characteristics of vehicular edge communications. We proved that the problem is bi-convex and proposed a solution based on the BCD method, denoted as \ac{VR-VFL}, which is guaranteed to reach a partially optimum solution. Our numerical results demonstrated that \ac{VR-VFL} achieves significant time savings compared to standard client selection methods in the literature. For future work, we plan to investigate a mixed-integer client selection mechanism instead of a bi-convex optimization problem, which can enforce more strictly the selection of vehicles, as well as vehicular-oriented learning tasks such as trajectory prediction  \vspace{-0.25cm}
\appendices \vspace{-0.35cm}
\section{Proof of Lemma~\ref{lemma:objectiveFunctionConvexity}} \label{appendix:objectiveFunctionConvexity}
Let us consider the objective function in Eq. \eqref{eq:objective2} as two separate terms. The second term is a pointwise maximum of convex functions, which is convex in $\mathbf{R}_t$. The first term includes a sum with a separate term for each vehicle. If we show that each term in the sum is convex, the sum will also be convex, thus completing the proof. For ease, we denote the first term from vehicle $v$ as \vspace{-0.1cm}
\begin{subequations}
\begin{align}
    \Theta_v =& \frac{\alpha D_v}{ u_{v,t}\left(1- \exp(\Xi_1) \exp(\Xi_2 (f)) \right)}, \\
    \Xi_1=& N_0 W_v / P_v L_v (1-\epsilon_v^2), \\
    \Xi_2(f)=& \frac{-|\hat{h}_{v}|^2 \epsilon_v^2}{(f-1) (1-\epsilon_v^2)}= \frac{-\Xi_3}{(f-1)}, \\
    \Xi_3 =& |\hat{h}_{v}|^2 \epsilon_v^2 / (1-\epsilon_v^2),
\end{align}
\end{subequations} \vspace{-0.0cm}
where $f(R_{v,t})=2^{R_{v,t}/W_v}$. Note that \( \Xi_1 + \Xi_2 < 0 \) due to the outage condition \( |\hat{h}_{v}|^2 > b_v \). Thus, \( 1 < f < 1 + \Xi_3 / \Xi_1 \), with the upper bound corresponding to the upper capacity limit when \( |\tilde{h}_{v}|^2 = 0 \), as seen in Eqs. \eqref{eq:sinr} and \eqref{eq:Cap}. Then, the second derivative of $\Theta_v$ with respect to $R_{v,t}$ is
\begin{align}
    \frac{\partial^2 \Theta_v}{\partial R_{v,t}^2} =& \left(\frac{\alpha D_v}{ u_{v,t} \left(1- \exp(\Xi_1) \exp(\Xi_2) \right)^3}\right) \times \nonumber \\
    & \Bigg( \left(\frac{\partial \Xi_2}{\partial R_{v,t}}\right)^2 \left( 1 + \exp(\Xi_1 + \Xi_2) \right) \nonumber \\
    & \quad + \frac{\partial^2 \Xi_2 }  {\partial R_{v,t}^2} \left( 1 - \exp(\Xi_1 + \Xi_2) \right) \Bigg). \label{eq:Hessian}
\end{align}

Then, we write the partial derivatives of $\Xi_2$ as
\begin{align}
    \frac{\partial \Xi_2}{\partial R_{v,t}} =& \Xi_3 \frac{\ln{2}}{W_v} \frac{f}{(f-1)^2} , \\
    \frac{\partial^2 \Xi_2}{\partial R_{v,t}^2} =&- \Xi_3 \frac{\ln^2{2}}{W_v^2} \frac{f(f+1)}{(f-1)^3}.
\end{align} \vspace{-0.2cm}

If the second derivative of $\Theta_v$ is positive, then $\Theta_v$ is convex. Thus, note that the first term in  parenthesis in Eq. \eqref{eq:Hessian} is always positive, then we only need to verify that the second term is positive. We can rewrite convexity condition of $\Theta_v$ as \vspace{-0.3cm}
\begin{align}
   \Lambda(f)= \frac{f}{f^2-1} \frac{\left( 1 + \exp(\Xi_1 + \Xi_2) \right)}{\left( 1 - \exp(\Xi_1 + \Xi_2) \right)} -\frac{1}{\Xi_3} >0,
\end{align}
where $\Lambda$ is obtained by multiplying the second term in parenthesis in Eq. \eqref{eq:Hessian} with the positive number $W_v^2 (f-1)^3/((\ln^2 2) \Xi_3^2 f (f+1) (1-\exp(\Xi_1+\Xi_2 )))$, which does not change its sign. Then, partial derivative of $\Lambda$ w.r.t. $f$ is
\begin{align} \nonumber
    \frac{\partial \Lambda(f)}{\partial f} =& \frac{-(f^2+1)}{(f^2-1)^2} \frac{1+\exp(\Xi_1+\Xi_2 )}{1-\exp(\Xi_1+\Xi_2 )}\\ 
    &- \frac{2 f \Xi_3 }{(f+1)(f-1)^3} \frac{(\exp(\Xi_1+\Xi_2 ))^2}{(1-\exp(\Xi_1+\Xi_2 ))^2}. \label{eq:partialLambda}
\end{align}

Since both terms in Eq. \eqref{eq:partialLambda} are negative, $\Lambda$ is decreasing with respect to $f$. Moreover, note that $\lim_{f \to 1+ \Xi_3/\Xi_1} \Lambda(f) \to \infty$. Therefore, $\Lambda$ is decreasing in $f$ and it is greater than zero at the upper limit of $f=1+ \Xi_3/\Xi_1$, making it positive for all $f>1$. This completes the proof.\vspace{-0.1cm}


\end{document}

%% file: AcronymList.tex
\begin{acronym}
    \acro{FL}{federated learning}
    \acro{FLoWN}{FL over wireless networks}
    \acro{AFL}{asynchronous FL}
    \acro{SGD}{stochastic gradient descent}
    \acro{FLoVN}{FL over vehicular networks}
    \acro{FLVEN}{FL in vehicular edge networks}
    \acro{VN}{vehicular networks}
    \acro{HFL}{hierarchical federated learning}
    \acro{RSU}{roadside unit}
    \acro{I-CSI}{imperfect channel state information}
    \acro{CSI}{channel state information}
    \acro{CDF}{cumulative distribution function}
    \acro{MMSE}{minimum mean square error}
    \acro{SINR}{signal-to-interference-plus-noise ratio}
    \acro{PDF}{probability density function}
    \acro{BCD}{block coordinate descent}
    \acro{VR-VFL}{variable rate - vehicular federated learning}
    \acro{non-IID}{non independent and identically distributed}
    \acro{IID}{independent and identically distributed}
\end{acronym}